\documentclass{llncs}
\usepackage{version}
\usepackage{algeqv}

\title{On Algorithmic Equivalence of Instruction Sequences for 
       Computing Bit String Functions}
\author{J.A. Bergstra \and C.A. Middelburg}
\institute{Informatics Institute, Faculty of Science, University of
           Amsterdam, \\
           Science Park~904, 1098~XH Amsterdam, the Netherlands \\
           \email{J.A.Bergstra@uva.nl,C.A.Middelburg@uva.nl}}

\begin{document}
\maketitle

\begin{abstract}
Every partial function from bit strings of a given length to bit strings 
of a possibly different given length can be computed by a finite 
instruction sequence that contains only instructions to set and get the 
content of Boolean registers, forward jump instructions, and a 
termination instruction.
We look for an equivalence relation on instruction sequences of this 
kind that captures to a reasonable degree the intuitive notion that two 
instruction sequences express the same algorithm.
\begin{keywords} 
structural algorithmic equivalence, 
structural computational equivalence, 
single-pass instruction sequence, 
bit string function.
\end{keywords}%
\begin{classcode}
F.1.1, F.2.0.
\end{classcode}
\end{abstract}

\section{Introduction}
\label{sect-intro}

In~\cite{BM13a}, it is among other things shown that a total function on 
bit strings whose result is a bit string of length $1$ belongs to P/poly 
iff it can be computed by polynomial-length instruction sequences that 
contain only instructions to set and get the content of Boolean 
registers, forward jump instructions, and a termination instruction.
In~\cite{BM13d}, where instruction sequences are considered which 
contain backward jump instructions in addition to the above-mentioned 
instructions, it is among other things shown that there exist a total 
function on bit strings and an algorithm for computing the function such 
that the function can be computed according to the algorithm by 
quadratic-length instruction sequences without backward jump 
instructions and by linear-length instruction sequences with backward 
jump instructions.

With that, we implicitly assumed that the instruction sequences without 
backward jump instructions concerned and the instruction sequences with 
backward jump instructions concerned express the same algorithm.
We considered this assumption acceptable because all the different views 
on what characterizes an algorithm lead to the conclusion that we have 
to do here with different expressions of the same algorithm.
However, we cannot prove this due to the absence of a mathematically 
precise definition of an equivalence relation on the instruction 
sequences of the kind considered that captures the intuitive notion that 
two instruction sequences express the same algorithm.

Attempts have been made to define such an equivalence relation in other
settings (see e.g.~\cite{Mil71a,Mos01b,Yan11a}), but there is still 
doubt if there exists an equivalence relation that completely captures 
the intuitive notion that two instruction sequences express the same 
algorithm (see e.g.~\cite{BDG09a,Dea07a}).
If such an equivalence relation would exist, algorithms could be defined 
as equivalence classes of programs with respect to this equivalence 
relation.
We take the viewpoint that there may be different degrees to which 
programs express the same algorithm, these different degrees may give 
rise to different equivalence relations, and these different equivalence 
relations may be interesting despite the fact that they incompletely 
capture the intuitive notion that two instruction sequences express the 
same algorithm.

In this paper, we look for an equivalence relation on instruction 
sequences of the kind considered in~\cite{BM13a} that captures to a 
reasonable degree the intuitive notion that two instruction sequences 
express the same algorithm.
Because the existing viewpoints on what is an algorithm are diverse in 
character and leave loose ends, there is little that we can build on.
Therefore, we restrict ourselves to what is virtually the simplest case.
That is, we take two fixed but arbitrary natural numbers $n,m$ and 
restrict ourselves to instruction sequences for computing partial 
functions from $\set{0,1}^n$ to $\set{0,1}^m$.
If $n$ and $m$ are very large, this simple case covers, at a very low 
level of program and data representation, many non-interactive programs 
that are found in actual practice.

In~\cite{BL02a}, an attempt is made to approach the semantics of 
programming languages from the perspective that a program is in essence 
an instruction sequence.
The groundwork for the approach is an algebraic theory of single-pass
instruction sequences, called program algebra, and an algebraic theory
of mathematical objects that represent the behaviours produced by
instruction sequences under execution, called basic thread algebra.%
\footnote
{In~\cite{BL02a}, basic thread algebra is introduced under the name
 basic polarized process algebra.
}
As in the previous works originating from this work on an approach to 
programming language semantics (see e.g.~\cite{BM12b}), the work 
presented in this paper is carried out in the setting of program algebra 
and basic thread algebra.
In this paper, we only give brief summaries of program algebra, basic 
thread algebra, and  the greater part of the extension of basic thread 
algebra that is used.
A comprehensive introduction, including examples, can among other things
be found in~\cite{BM12b}.

This paper is organized as follows.
First, we give a survey of program algebra and basic thread algebra 
(Section~\ref{sect-PGA-and-BTA}) and a survey of the greater part of the 
extension of basic thread algebra that is used in this paper 
(Section~\ref{sect-TSI}).
Next, we introduce the remaining part of the extension of basic thread 
algebra that is used in this paper (Section~\ref{sect-tracking-use}) and 
present the instruction sequences that concern us in this paper 
(Section~\ref{sect-Boolean-register}).
After that, we give some background on the intuitive notion that two 
instruction sequences express the same algorithm 
(Section~\ref{sect-background-algeqv}) and a picture of our intuition 
about this notion (Section~\ref{sect-intuition-algeqv}).
Then, we define an algorithmic equivalence relation on the instruction 
sequences introduced before that corresponds to this intuition 
(Section~\ref{sect-definition-saeqv}) and a related equivalence relation 
that happens to be too coarse to be an algorithmic equivalence relation 
(Section~\ref{sect-sceqv}). 
Following this, we show how the algorithmic equivalence relation 
introduced before can be lifted to programs in a higher-level program 
notation (Section~\ref{sect-higher-level-algeqv}) and point out that we 
are still far from the definitive answer to the question ``what is an 
algorithm?'' (Section~\ref{sect-discussion}).
Finally, we make some concluding remarks (Section~\ref{sect-concl}).

Henceforth, we will regularly refer to the intuitive notion that two 
instruction sequences express the same algorithm as the notion of 
algorithmic sameness.
Moreover, we will use the term algorithmic equivalence relation as a
general name for equivalence relations that capture to some degree this
intuitive notion.

The groundwork for the work presented in the current paper has been laid 
in previous papers.
The portions of this groundwork that are necessary to understand the 
current paper have been copied near verbatim or slightly modified.
The greater part of Sections~\ref{sect-PGA-and-BTA} 
and~\ref{sect-Boolean-register} originate from~\cite{BM13a} and 
Section~\ref{sect-TSI} originates from~\cite{BM09k}.

\section{Program Algebra and Basic Thread Algebra}
\label{sect-PGA-and-BTA}

In this section, we give a survey of \PGA\ (ProGram Algebra) and \BTA\ 
(Basic Thread Algebra) and make precise in the setting of \BTA\ which 
behaviours are produced by the instruction sequences considered in \PGA\
under execution.
The greater part of this section originates from~\cite{BM13a}.

In \PGA, it is assumed that there is a fixed but arbitrary set $\BInstr$
of \emph{basic instructions}.
The intuition is that the execution of a basic instruction may modify a 
state and produces a reply at its completion.
The possible replies are $\False$ and $\True$.
The actual reply is generally state-dependent.
The set $\BInstr$ is the basis for the set of instructions that may 
occur in the instruction sequences considered in \PGA.
The elements of the latter set are called \emph{primitive instructions}.
There are five kinds of primitive instructions:
\begin{itemize}
\item
for each $a \in \BInstr$, a \emph{plain basic instruction} $a$;
\item
for each $a \in \BInstr$, a \emph{positive test instruction} $\ptst{a}$;
\item
for each $a \in \BInstr$, a \emph{negative test instruction} $\ntst{a}$;
\item
for each $l \in \Nat$, a \emph{forward jump instruction} $\fjmp{l}$;
\item
a \emph{termination instruction} $\halt$.
\end{itemize}
We write $\PInstr$ for the set of all primitive instructions.

On execution of an instruction sequence, these primitive instructions
have the following effects:
\begin{itemize}
\item
the effect of a positive test instruction $\ptst{a}$ is that basic
instruction $a$ is executed and execution proceeds with the next
primitive instruction if $\True$ is produced and otherwise the next
primitive instruction is skipped and execution proceeds with the
primitive instruction following the skipped one --- if there is no
primitive instruction to proceed with,
inaction occurs;
\item
the effect of a negative test instruction $\ntst{a}$ is the same as
the effect of $\ptst{a}$, but with the role of the value produced
reversed;
\item
the effect of a plain basic instruction $a$ is the same as the effect
of $\ptst{a}$, but execution always proceeds as if $\True$ is produced;
\item
the effect of a forward jump instruction $\fjmp{l}$ is that execution
proceeds with the $l$th next primitive instruction --- if $l$ equals $0$ 
or there is no primitive instruction to proceed with, inaction occurs;
\item
the effect of the termination instruction $\halt$ is that execution 
terminates.
\end{itemize}

\PGA\ has one sort: the sort $\InSeq$ of \emph{instruction sequences}. 
We make this sort explicit to anticipate the need for many-sortedness
later on.
To build terms of sort $\InSeq$, \PGA\ has the following constants and 
operators:
\begin{itemize}
\item
for each $u \in \PInstr$, 
the \emph{instruction} constant $\const{u}{\InSeq}$\,;
\item
the binary \emph{concatenation} operator 
$\funct{\ph \conc \ph}{\InSeq \x \InSeq}{\InSeq}$\,;
\item
the unary \emph{repetition} operator 
$\funct{\ph\rep}{\InSeq}{\InSeq}$\,.
\end{itemize}
Terms of sort $\InSeq$ are built as usual in the one-sorted case.
We assume that there are infinitely many variables of sort $\InSeq$, 
including $X,Y,Z$.
We use infix notation for concatenation and postfix notation for
repetition.

A closed \PGA\ term is considered to denote a non-empty, finite or
eventually periodic infinite sequence of primitive instructions.%
\footnote
{An eventually periodic infinite sequence is an infinite sequence with
 only finitely many distinct suffixes.}
The instruction sequence denoted by a closed term of the form
$t \conc t'$ is the instruction sequence denoted by $t$
concatenated with the instruction sequence denoted by $t'$.
The instruction sequence denoted by a closed term of the form $t\rep$
is the instruction sequence denoted by $t$ concatenated infinitely
many times with itself.

Closed \PGA\ terms are considered equal if they represent the same
instruction sequence.
The axioms for instruction sequence equivalence are given in
Table~\ref{axioms-PGA}.%
\begin{table}[!t]
\caption{Axioms of \PGA}
\label{axioms-PGA}
\begin{eqntbl}
\begin{axcol}
(X \conc Y) \conc Z = X \conc (Y \conc Z)              & \axiom{PGA1} \\
(X^n)\rep = X\rep                                      & \axiom{PGA2} \\
X\rep \conc Y = X\rep                                  & \axiom{PGA3} \\
(X \conc Y)\rep = X \conc (Y \conc X)\rep              & \axiom{PGA4}
\end{axcol}
\end{eqntbl}
\end{table}
In this table, $n$ stands for an arbitrary natural number greater than
$0$.
For each $n > 0$, the term $t^n$, where $t$ is a \PGA\ term, is defined 
by induction on $n$ as follows: $t^1 = t$ and $t^{n+1} = t \conc t^n$.

A typical model of \PGA\ is the model in which:
\begin{itemize}
\item
the domain is the set of all finite and eventually periodic infinite
sequences over the set $\PInstr$ of primitive instructions;
\item
the operation associated with ${} \conc {}$ is concatenation;
\item
the operation associated with ${}\rep$ is the operation ${}\srep$
defined as follows:
\begin{itemize}
\item
if $U$ is finite, then $U\srep$ is the unique infinite sequence $U'$ 
such that $U$ concatenated $n$ times with itself is a proper prefix of 
$U'$ for each $n \in \Nat$;
\item
if $U$ is infinite, then $U\srep$ is $U$.
\pagebreak[2]
\end{itemize}
\end{itemize}
We confine ourselves to this model of \PGA, which is an initial model of 
\PGA, for the interpretation of \PGA\ terms.
In the sequel, we use the term \emph{PGA instruction sequence} for the 
elements of the domain of this model, and we denote the interpretations
of the constants and operators in this model by the constants and 
operators themselves.
Below, we will use \BTA\ to make precise which behaviours are produced 
by \PGA\ instruction sequences under execution.

In \BTA, it is assumed that a fixed but arbitrary set $\BAct$ of
\emph{basic actions} has been given.
The objects considered in \BTA\ are called threads.
A thread represents a behaviour which consists of performing basic 
actions in a sequential fashion.
Upon each basic action performed, a reply from an execution environment
determines how the thread proceeds.
The possible replies are the values $\False$ and $\True$.

\BTA\ has one sort: the sort $\Thr$ of \emph{threads}. 
We make this sort explicit to anticipate the need for many-sortedness
later on.
To build terms
of sort $\Thr$, \BTA\ has the following constants and operators:
\begin{itemize}
\item
the \emph{inaction} constant $\const{\DeadEnd}{\Thr}$;
\item
the \emph{termination} constant $\const{\Stop}{\Thr}$;
\item
for each $a \in \BAct$, the binary \emph{postconditional composition} 
operator $\funct{\pcc{\ph}{a}{\ph}}{\Thr \x \Thr}{\Thr}$.
\end{itemize}
Terms of sort $\Thr$ are built as usual in the one-sorted case. 
We assume that there are infinitely many variables of sort $\Thr$, 
including $x,y$.
We use infix notation for postconditional composition. 
We introduce \emph{basic action prefixing} as an abbreviation: 
$a \bapf t$, where $t$ is a \BTA\ term, abbreviates 
$\pcc{t}{a}{t}$.
We identify expressions of the form $a \bapf t$ with the \BTA\
term they stand for.

The thread denoted by a closed term of the form $\pcc{t}{a}{t'}$
will first perform $a$, and then proceed as the thread denoted by
$t$ if the reply from the execution environment is $\True$ and proceed
as the thread denoted by $t'$ if the reply from the execution
environment is $\False$. 
The thread denoted by $\Stop$ will do no more than terminate and the 
thread denoted by $\DeadEnd$ will become inactive.

Closed \BTA\ terms are considered equal if they are syntactically the
same.
Therefore, \BTA\ has no axioms.

Each closed \BTA\ term denotes a finite thread, i.e.\ a thread with a
finite upper bound to the number of basic actions that it can perform.
Infinite threads, i.e.\ threads without a finite upper bound to the
number of basic actions that it can perform, can be defined by means of 
a set of recursion equations (see e.g.~\cite{BM09k}).
We are only interested in models of \BTA\ in which sets of recursion 
equations have unique solutions, such as the projective limit model 
of \BTA\ presented in~\cite{BB03a}.
We confine ourselves to this model of \BTA, which has an initial model 
of \BTA\ as a submodel, for the interpretation of \BTA\ terms. 
In the sequel, we use the term \emph{BTA thread} or simply \emph{thread} 
for the elements of the domain of this model, and we denote the 
interpretations of the constants and operators in this model by the 
constants and operators themselves.

Regular threads, i.e.\ finite or infinite threads that can only be in a 
finite number of states, can be defined by means of a finite set of 
recursion equations.
The behaviours produced by \PGA\ instruction sequences under execution 
are exactly the behaviours represented by regular threads, with the basic 
instructions taken for basic actions.
The behaviours produced by finite \PGA\ instruction sequences are the 
behaviours represented by finite threads.

We combine \PGA\ with \BTA\ and extend the combination with
the \emph{thread extraction} operator $\funct{\extr{\ph}}{\InSeq}{\Thr}$, 
the axioms given in Table~\ref{axioms-thread-extr},%
\begin{table}[!tb]
\caption{Axioms for the thread extraction operator}
\label{axioms-thread-extr}
\begin{eqntbl}
\begin{eqncol}
\extr{a} = a \bapf \DeadEnd \\
\extr{a \conc X} = a \bapf \extr{X} \\
\extr{\ptst{a}} = a \bapf \DeadEnd \\
\extr{\ptst{a} \conc X} =
\pcc{\extr{X}}{a}{\extr{\fjmp{2} \conc X}} \\
\extr{\ntst{a}} = a \bapf \DeadEnd \\
\extr{\ntst{a} \conc X} =
\pcc{\extr{\fjmp{2} \conc X}}{a}{\extr{X}}
\end{eqncol}
\qquad
\begin{eqncol}
\extr{\fjmp{l}} = \DeadEnd \\
\extr{\fjmp{0} \conc X} = \DeadEnd \\
\extr{\fjmp{1} \conc X} = \extr{X} \\
\extr{\fjmp{l+2} \conc u} = \DeadEnd \\
\extr{\fjmp{l+2} \conc u \conc X} = \extr{\fjmp{l+1} \conc X} \\
\extr{\halt} = \Stop \\
\extr{\halt \conc X} = \Stop
\end{eqncol}
\end{eqntbl}
\end{table}
and the rule that $\extr{X} = \DeadEnd$ if $X$ has an infinite chain of 
forward jumps beginning at its first primitive instruction.%
\footnote
{This rule, which can be formalized using an auxiliary structural 
congruence predicate (see e.g.~\cite{BM07g}), is unnecessary when 
considering only finite \PGA\ instruction sequences.
}
In Table~\ref{axioms-thread-extr}, $a$ stands for an arbitrary basic 
instruction from $\BInstr$, $u$ stands for an arbitrary primitive 
instruction from $\PInstr$, and $l$ stands for an arbitrary natural 
number from $\Nat$.
For each closed \PGA\ term $t$, $\extr{t}$ denotes the behaviour  
produced by the instruction sequence denoted by $t$ under execution.

Equality of \PGA\ instruction sequence as axiomatized by the axioms of 
\PGA\ is extensional equality: two \PGA\ instruction sequences are equal 
if they have the same length and the $n$th instructions are equal for 
all $n > 0$ that are less than or equal to the common length.
We define the function $\len$ that assigns to each \PGA\ instruction 
sequence its length:
\begin{ldispl}
\len(u) = 1\;, \\
\len(X \conc Y) = \len(X) + \len(Y)\;, \\
\len(X\rep) = \omega\;, 
\end{ldispl}%
and we define for each $n > 0$ the function $i_n$ that assigns to each 
\PGA\ instruction sequence its $n$th instruction if $n$ is less than or 
equal to its length and $\fjmp{0}$ otherwise:
\begin{ldispl}
i_1(u) = u\;, \\
i_1(u \conc X) = u\;, \\
i_{n+1}(u) = \fjmp{0}\;, \\
i_{n+1}(u \conc X) = i_n(X)\;.
\end{ldispl}%
Let $X$ and $Y$ be \PGA\ instruction sequences.
Then we have by extensionality that 
\begin{ldispl}
\mbox{$X = Y$ iff 
      $\len(X) = \len(Y)$ and $i_n(X) = i_n(Y)$ for all $n > 0$}\;.
\end{ldispl}%
This means that each \PGA\ instruction sequence $X$ is uniquely 
characterized by $\len(X)$ and $i_n(X)$ for all $n > 0$, which will be 
used several times in Section~\ref{sect-definition-saeqv}.

The depth of a finite thread is the maximum number of basic actions that 
the thread can perform before it terminates or becomes inactive.
We define the function $\depth$ that assigns to each finite BTA\ thread 
its depth:
\begin{ldispl}
\depth(\Stop) = 0\;, \\
\depth(\DeadEnd) = 0\;, \\
\depth(\pcc{x}{a}{y}) = \max \set{\depth(x),\depth(y)} + 1\;.
\end{ldispl}%
Let $x$ be a thread of the form $a_1 \bapf \ldots \bapf a_n \bapf \Stop$ 
or the form $a_1 \bapf \ldots \bapf a_n \bapf \DeadEnd$.
Then $\depth(x)$ represents the only number of basic actions that $x$ can 
perform before it terminates or becomes inactive.
This will be used in Section~\ref{sect-definition-saeqv}.

\section{Interaction of Threads with Services}
\label{sect-TSI}

Services are objects that represent the behaviours exhibited by 
components of execution environments of instruction sequences at a high 
level of abstraction.
A service is able to process certain methods.
The processing of a method may involve a change of the service.
At completion of the processing of a method, the service produces a
reply value.
Execution environments are considered to provide a family of 
uniquely-named services.
A thread may interact with the named services from the service family 
provided by an execution environment.
That is, a thread may perform a basic action for the purpose of 
requesting a named service to process a method and to return a reply 
value at completion of the processing of the method.
In this section, we extend \BTA\ with services, service families, a 
composition operator for service families, and operators that are 
concerned with this kind of interaction.
This section originates from~\cite{BM09k}.

In \SFA, the algebraic theory of service families introduced below, it 
is assumed that a fixed but arbitrary set $\Meth$ of \emph{methods} has 
been given.
Moreover, the following is assumed with respect to services:
\begin{itemize}
\item
a signature $\Sig{\Services}$ has been given that includes the following
sorts:
\begin{itemize}
\item
the sort $\Serv$ of
\emph{services};
\item
the sort $\Repl$ of \emph{replies};
\end{itemize}
and the following constants and operators:
\begin{itemize}
\item
the
\emph{empty service} constant $\const{\emptyserv}{\Serv}$;
\item
the \emph{reply} constants $\const{\False,\True,\Div}{\Repl}$;
\item
for each $m \in \Meth$, the
\emph{derived service} operator $\funct{\derive{m}}{\Serv}{\Serv}$;
\item
for each $m \in \Meth$, the
\emph{service reply} operator $\funct{\sreply{m}}{\Serv}{\Repl}$;
\end{itemize}
\item
a minimal $\Sig{\Services}$-algebra $\ServAlg$ has been given in which
$\False$, $\True$, and $\Div$ are mutually different, and
\begin{itemize}
\item
$\LAND{m \in \Meth}{}
  \derive{m}(z) = z \Land \sreply{m}(z) = \Div \Limpl
  z = \emptyserv$
holds;
\item
for each $m \in \Meth$,
$\derive{m}(z) = \emptyserv \Liff \sreply{m}(z) = \Div$ holds.
\end{itemize}
\end{itemize}

The intuition concerning $\derive{m}$ and $\sreply{m}$ is that on a
request to service $s$ to process method $m$:
\begin{itemize}
\item
if $\sreply{m}(s) \neq \Div$, $s$ processes $m$, produces the reply
$\sreply{m}(s)$, and then proceeds as $\derive{m}(s)$;
\item
if $\sreply{m}(s) = \Div$, $s$ is not able to process method $m$ and
proceeds as $\emptyserv$.
\end{itemize}
The empty service $\emptyserv$ itself is unable to process any method.

It is also assumed that a fixed but arbitrary set $\Foci$ of
\emph{foci} has been given.
Foci play the role of names of services in a service family. 

\SFA\ has the sorts, constants and operators from $\Sig{\Services}$ and
in addition the sort $\ServFam$ of \emph{service families} and the 
following constant and operators:
\begin{itemize}
\item
the
\emph{empty service family} constant $\const{\emptysf}{\ServFam}$;
\item
for each $f \in \Foci$, the unary
\emph{singleton service family} operator
$\funct{\mathop{f{.}} \ph}{\Serv}{\ServFam}$;
\item
the binary
\emph{service family composition} operator
$\funct{\ph \sfcomp \ph}{\ServFam \x \ServFam}{\ServFam}$;
\item
for each $F \subseteq \Foci$, the unary
\emph{encapsulation} operator $\funct{\encap{F}}{\ServFam}{\ServFam}$.
\end{itemize}
We assume that there are infinitely many variables of sort $\Serv$,
including $z$, and infinitely many variables of sort $\ServFam$,
including $u,v,w$.
Terms are built as usual in the many-sorted case
(see e.g.~\cite{ST99a,Wir90a}).
We use prefix notation for the singleton service family operators and
infix notation for the service family composition operator.
We write $\Sfcomp{i = 1}{n} t_i$, where $t_1,\ldots,t_n$ are
terms of sort $\ServFam$, for the term
$t_1 \sfcomp \ldots \sfcomp t_n$.

The service family denoted by $\emptysf$ is the empty service family.
The service family denoted by a closed term of the form $f.t$ consists 
of one named service only, the service concerned is the service denoted 
by $t$, and the name of this service is $f$.
The service family denoted by a closed term of the form
$t \sfcomp t'$ consists of all named services that belong to either the
service family denoted by $t$ or the service family denoted by $t'$.
In the case where a named service from the service family denoted by
$t$ and a named service from the service family denoted by $t'$ have
the same name, they collapse to an empty service with the name
concerned.
The service family denoted by a closed term of the form $\encap{F}(t)$ 
consists of all named services with a name not in $F$ that belong to
the service family denoted by $t$.

The axioms of \SFA\ are given in 
Table~\ref{axioms-SFA}.%
\begin{table}[!t]
\caption{Axioms of \SFA}
\label{axioms-SFA}
{
\begin{eqntbl}
\begin{axcol}
u \sfcomp \emptysf = u                                 & \axiom{SFC1} \\
u \sfcomp v = v \sfcomp u                              & \axiom{SFC2} \\
(u \sfcomp v) \sfcomp w = u \sfcomp (v \sfcomp w)      & \axiom{SFC3} \\
f.z \sfcomp f.z' = f.\emptyserv                  & \axiom{SFC4}
\end{axcol}
\qquad
\begin{saxcol}
\encap{F}(\emptysf) = \emptysf                     & & \axiom{SFE1} \\
\encap{F}(f.z) = \emptysf & \mif f \in F       & \axiom{SFE2} \\
\encap{F}(f.z) = f.z    & \mif f \notin F    & \axiom{SFE3} \\
\multicolumn{2}{@{}l@{\quad}}
 {\encap{F}(u \sfcomp v) =
  \encap{F}(u) \sfcomp \encap{F}(v)}               & \axiom{SFE4}
\end{saxcol}
\end{eqntbl}
}
\end{table}
In this table, $f$ stands for an arbitrary focus from $\Foci$ and
$F$ stands for an arbitrary subset of $\Foci$.
These axioms simply formalize the informal explanation given
above.

For the set $\BAct$ of basic actions, we now take the set
$\set{f.m \where f \in \Foci, m \in \Meth}$.
Performing a basic action $f.m$ is taken as making a request to the
service named $f$ to process method $m$.

We combine \BTA\ with \SFA\ and extend the combination with the 
following constants and operators: 
\pagebreak[2]
\begin{itemize}
\item
the binary \emph{abstracting use} operator
$\funct{\ph \sfause \ph}{\Thr \x \ServFam}{\Thr}$;
\item
the binary \emph{apply} operator
$\funct{\ph \sfapply \ph}{\Thr \x \ServFam}{\ServFam}$;
\end{itemize}
and the axioms given in Tables~\ref{axioms-abstracting-use} 
and~\ref{axioms-apply}.%
\begin{table}[!t]
\caption{Axioms for the abstracting use operator}
\label{axioms-abstracting-use}
\begin{eqntbl}
\begin{saxcol}
\Stop  \sfause u = \Stop                              & & \axiom{AU1} \\
\DeadEnd \sfause u = \DeadEnd                         & & \axiom{AU2} \\
(\pcc{x}{f.m}{y}) \sfause \encap{\set{f}}(u) =
\pcc{(x \sfause \encap{\set{f}}(u))}
 {f.m}{(y \sfause \encap{\set{f}}(u))}                & & \axiom{AU3} \\
(\pcc{x}{f.m}{y}) \sfause (f.t \sfcomp \encap{\set{f}}(u)) =
x \sfause (f.\derive{m}t \sfcomp \encap{\set{f}}(u))
                          & \mif \sreply{m}(t) = \True  & \axiom{AU4} \\
(\pcc{x}{f.m}{y}) \sfause (f.t \sfcomp \encap{\set{f}}(u)) =
y \sfause (f.\derive{m}t \sfcomp \encap{\set{f}}(u))
                          & \mif \sreply{m}(t) = \False & \axiom{AU5} \\
(\pcc{x}{f.m}{y}) \sfause (f.t \sfcomp \encap{\set{f}}(u)) = \DeadEnd
                          & \mif \sreply{m}(t) = \Div   & \axiom{AU6}
\end{saxcol}
\end{eqntbl}
\end{table}
\begin{table}[!t]
\caption{Axioms for the apply operator}
\label{axioms-apply}
\begin{eqntbl}
\begin{saxcol}
\Stop  \sfapply u = u                                  & & \axiom{A1} \\
\DeadEnd \sfapply u = \emptysf                         & & \axiom{A2} \\
(\pcc{x}{f.m}{y}) \sfapply \encap{\set{f}}(u) = \emptysf
                                                       & & \axiom{A3} \\
(\pcc{x}{f.m}{y}) \sfapply (f.t \sfcomp \encap{\set{f}}(u)) =
x \sfapply (f.\derive{m}t \sfcomp \encap{\set{f}}(u))
                           & \mif \sreply{m}(t) = \True  & \axiom{A4} \\
(\pcc{x}{f.m}{y}) \sfapply (f.t \sfcomp \encap{\set{f}}(u)) =
y \sfapply (f.\derive{m}t \sfcomp \encap{\set{f}}(u))
                           & \mif \sreply{m}(t) = \False & \axiom{A5} \\
(\pcc{x}{f.m}{y}) \sfapply (f.t \sfcomp \encap{\set{f}}(u)) = \emptysf
                           & \mif \sreply{m}(t) = \Div   & \axiom{A6}
\end{saxcol}
\end{eqntbl}
\end{table}
In these tables, $f$ stands for an arbitrary focus from $\Foci$, $m$ 
stands for an arbitrary method from $\Meth$, and $t$ stands for an 
arbitrary term of sort $\Serv$.
The axioms formalize the informal explanation given below and in 
addition stipulate what is the result of abstracting use and apply if 
inappropriate foci or methods are involved.
We use infix notation for the abstracting use and apply operators.

The thread denoted by a closed term of the form $t \sfause t'$ and the
service family denoted by a closed term of the form $t \sfapply t'$ are
the thread and service family, respectively, that result from processing
the method of each basic action performed by the thread denoted by $t$
by the service in the service family denoted by $t'$ with the focus
of the basic action as its name if such a service exists.
When the method of a basic action performed by a thread is processed by
a service, the service changes in accordance with the method concerned
and the thread reduces to one of the two threads that it can possibly 
proceed with dependent on the reply value produced by the service.

\section{The Tracking Use Operator}
\label{sect-tracking-use}

In this section, we extend the combination of \BTA\ with \SFA\ further
with a tracking use operator.
Abstracting use does not leave a trace of what has taken place during 
the interaction between thread and service family, whereas tracking use 
leaves a detailed trace.
The tracking use operator has been devised for its usefulness for the 
work presented in this paper.

For the set $\BAct$ of basic actions, we now take the set
$\set{f.m \where f \in \Foci, m \in \Meth} \union
 \set{\iact{f.m}{r} \where 
      f \in \Foci,\linebreak[2] m \in \Meth, 
      r \in \set{\False,\True,\Div}}$.
A basic action $\iact{f.m}{r}$ represents the successful handling of a 
request to the service named $f$ to process method $m$ with reply $r$ if 
$r \neq \Div$, and the unsuccessful handling of such a request 
otherwise.

We extend the combination of \BTA\ with \SFA\ further with the 
following operator:
\begin{itemize}
\item
the binary \emph{tracking use} operator
$\funct{\ph \sfpause \ph}{\Thr \x \ServFam}{\Thr}$;
\end{itemize}
and the axioms given in Tables~\ref{axioms-tracking-use}.%
\begin{table}[!t]
\caption{Axioms for the tracking use operator}
\label{axioms-tracking-use}
\begin{eqntbl}
\begin{saxcol}
\Stop  \sfpause u = \Stop                              & & \axiom{PAU1} \\
\DeadEnd \sfpause u = \DeadEnd                         & & \axiom{PAU2} \\
(\pcc{x}{f.m}{y}) \sfpause \encap{\set{f}}(u) =
\pcc{(x \sfpause \encap{\set{f}}(u))}
 {f.m}{(y \sfpause \encap{\set{f}}(u))}                & & \axiom{PAU3} \\
(\pcc{x}{f.m}{y}) \sfpause (f.t \sfcomp \encap{\set{f}}(u)) =
\iact{f.m}{\True} \bapf 
(x \sfpause (f.\derive{m}t \sfcomp \encap{\set{f}}(u)))
                          & \mif \sreply{m}(t) = \True  & \axiom{PAU4} \\
(\pcc{x}{f.m}{y}) \sfpause (f.t \sfcomp \encap{\set{f}}(u)) =
\iact{f.m}{\False} \bapf 
(y \sfpause (f.\derive{m}t \sfcomp \encap{\set{f}}(u)))
                          & \mif \sreply{m}(t) = \False & \axiom{PAU5} \\
(\pcc{x}{f.m}{y}) \sfpause (f.t \sfcomp \encap{\set{f}}(u)) = 
\iact{f.m}{\Div} \bapf \DeadEnd
                          & \mif \sreply{m}(t) = \Div   & \axiom{PAU6} \\
(\pcc{x}{\iact{f.m}{r}}{y}) \sfpause u = 
\iact{f.m}{r} \bapf (x \sfpause u)                    & & \axiom{PAU7} 
\eqnsep

(\pcc{x}{\iact{f.m}{r}}{y}) \sfause u = 
\iact{f.m}{r} \bapf (x \sfause u)                     & & \axiom{AU7} \\
(\pcc{x}{\iact{f.m}{r}}{y}) \sfapply u = x \sfapply u & & \axiom{A7} 
\end{saxcol}
\end{eqntbl}
\end{table}
In this table, $f$ stands for an arbitrary focus from $\Foci$, $m$ 
stands for an arbitrary method from $\Meth$, $r$ stands for an arbitrary
constant of sort $\Repl$, and $t$ stands for an arbitrary term of sort 
$\Serv$.
The axioms PAU1--PAU7 formalize the informal explanation given below.
The axioms AU7 and A7 stipulate what is the result of abstracting use 
and apply if basic actions of the form $\iact{f.m}{r}$ are involved.
We use infix notation for the tracking use operator.

The thread denoted by a closed term of the form $t \sfpause t'$ differs
from the thread denoted by a closed term of the form $t \sfause t'$ as 
follows: when the method of a basic action performed by a thread is 
processed by a service, the thread reduces to one of the two threads 
that it can possibly proceed with dependent on the reply value produced 
by the service prefixed by the basic action $\iact{f.m}{r}$, where $f$ 
is the name of the processing service, $m$ is the method processed, and 
$r$ is the reply value produced.
Thus, the resulting thread represents a trace of what has taken place.
Performing the basic action $\iact{f.m}{r}$ never leads to a change of 
any service and always leads to the reply $\True$. 

\section{Instruction Sequences Acting on Boolean Registers}
\label{sect-Boolean-register}

The basic instructions that concern us in the remainder of this paper 
are instructions to set and get the content of Boolean registers.
We describe in this section services that make up Boolean registers, 
introduce special foci that serve as names of Boolean registers, and 
describe the instruction sequences on which we will define an 
algorithmic equivalence relation.
The greater part of this section originates from~\cite{BM13a}.

First, we describe services that make up Boolean registers.
It is assumed that $\setbr{\False}, \setbr{\True}, \getbr \in \Meth$.
These methods are the ones that Boolean register services are able to 
process.
They can be explained as follows:
\begin{itemize}
\item
$\setbr{\False}$\,:
the contents of the Boolean register becomes $\False$ and the reply is
$\False$;
\item
$\setbr{\True}$\,:
the contents of the Boolean register becomes $\True$ and the reply is
$\True$;
\item
$\getbr$\,:
nothing changes and the reply is the contents of the Boolean register.
\end{itemize}

For $\Sig{\Services}$, we take the signature that consists of the sorts,
constants and operators that are mentioned in the assumptions with
respect to services made in Section~\ref{sect-TSI} and constants 
$\BR_0$ and $\BR_1$.

For $\ServAlg$, we take a minimal $\Sig{\Services}$-algebra that
satisfies the conditions that are mentioned in the assumptions with
respect to services made in Section~\ref{sect-TSI} and the following
conditions for each $b \in \Bool$:
\begin{ldispl}
\begin{geqns}
\derive{\setbr{\False}}(\BR_{b}) = \BR_{\False}\;,
\\[.5ex]
\derive{\setbr{\True}}(\BR_{b})  = \BR_{\True}\;,
\eqnsep 
\sreply{\setbr{\False}}(\BR_{b}) = \False\;,
\\
\sreply{\setbr{\True}}(\BR_{b})  = \True\;,
\end{geqns}
\qquad
\begin{gceqns}
\derive{\getbr}(\BR_{b}) = \BR_{b}\;,
\\[.5ex]
\derive{m}(\BR_{b}) = \emptyserv
 & \mif m \notin \set{\setbr{\False},\setbr{\True},\getbr}\;,
\eqnsep 
\sreply{\getbr}(\BR_{b}) = b\;,
\\
\sreply{m}(\BR_{b}) = \Div
 & \mif m \notin \set{\setbr{\False},\setbr{\True},\getbr}\;.
\end{gceqns}
\end{ldispl}%

In the instruction sequences which concern us in the remainder of this
paper, a number of Boolean registers is used as input registers, a
number of Boolean registers is used as auxiliary registers, and a number
of Boolean register is used as output register.
It is assumed that, for each $i \in \Natpos$, 
$\inbr{i}, \auxbr{i}, \outbr{i} \in \Foci$.
These foci play special roles:
\begin{itemize}
\item
for each $i \in \Natpos$, $\inbr{i}$ serves as the name of the Boolean
register that is used as $i$th input register in instruction sequences;
\item
for each $i \in \Natpos$, $\auxbr{i}$ serves as the name of the Boolean
register that is used as $i$th auxiliary register in instruction
sequences;
\item
for each $i \in \Natpos$, $\outbr{i}$ serves as the name of the Boolean 
register that is used as $i$th output register in instruction sequences.
\end{itemize}
We define the following sets:
\begin{ldispl}
\Fociin{n} = \set{\inbr{i} \where 1 \leq i \leq n}\;, \;
\Fociaux = \set{\auxbr{i} \where i \geq 1}\;, \;
\Fociout{n} = \set{\outbr{i} \where 1 \leq i \leq n}\;, \\
\Focibr{n}{m} = \Fociin{n} \union \Fociaux \union \Fociout{m}\;, \;
\Methbr = \set{\setbr{0},\setbr{1},\getbr}\;, \\
\BIbr{n}{m} =
 \set{f.\getbr    \where f \in \Fociin{n} \union \Fociaux} \union
 \set{f.\setbr{b} \where f \in \Fociaux \union \Fociout{m} \Land
                         b \in \Bool}\;.
\end{ldispl}%
We write
$\PIbr{n}{m}$ for the set of primitive instructions in the case where 
$\BIbr{n}{m}$ is taken for the set $\BInstr$ of basic instructions, and 
we write $\ISbr{n}{m}$ for the set of all finite \PGA\ instruction 
sequences in the case where $\PIbr{n}{m}$ is taken for the set $\PInstr$ 
of primitive instructions.
Moreover, we write $\Tbr{n}{m}$ for the set of all finite \BTA\ threads 
in the case where $\BIbr{n}{m}$ is taken for the set $\BAct$ of basic 
actions. 

Let $n,m \in \Nat$, let $\pfunct{f}{\set{0,1}^n}{\set{0,1}^m}$,%
\footnote
{We write as usual $\pfunct{f}{\set{0,1}^n}{\set{0,1}^m}$ to indicate
 that $f$ is partial function from $\set{0,1}^n$ to $\set{0,1}^m$.}
and let $X \in \ISbr{n}{m}$.
Then $X$ \emph{computes} $f$ if there exists an $l \in \Nat$ such that, 
for all $b''_1,\ldots,b''_l \in \Bool$:
\begin{itemize}
\item
for all $b_1,\ldots,b_n,b'_1,\ldots,b'_m \in \Bool$ with
$f(b_1,\ldots,b_n) = b'_1,\ldots,b'_m$:
\begin{ldispl}
(\extr{X} \sfause
  ((\Sfcomp{i=1}{n} \inbr{i}.\BR_{b_i}) \sfcomp
   (\Sfcomp{i=1}{l} \auxbr{i}.\BR_{b''_i}))) \sfapply
  (\Sfcomp{i=1}{m} \outbr{i}.\BR_{\False}) 
\\ \quad {} = 
\Sfcomp{i=1}{m} \outbr{i}.\BR_{b'_i}\;;
\end{ldispl}%
\item
for all $b_1,\ldots,b_n \in \Bool$ with
$f(b_1,\ldots,b_n)$ undefined:
\begin{ldispl}
(\extr{X} \sfause
  ((\Sfcomp{i=1}{n} \inbr{i}.\BR_{b_i}) \sfcomp
   (\Sfcomp{i=1}{l} \auxbr{i}.\BR_{b''_i}))) \sfapply
  (\Sfcomp{i=1}{m} \outbr{i}.\BR_{\False}) 
\\ \quad {} = 
\emptysf\;.
\end{ldispl}%
\end{itemize}
We know from Theorem~1 in~\cite{BM13a} that, for each $n,m \in \Nat$, 
for each $\funct{f}{\set{0,1}^n}{\set{0,1}^m}$, there exists an 
$X \in \ISbr{n}{m}$ such that $X$ computes $f$.
It is easy to see from the proof of the theorem that this result 
generalizes from total functions to partial functions.

\section{Background on the Notion of Algorithmic Sameness}
\label{sect-background-algeqv}

In section~\ref{sect-definition-saeqv}, we will define an equivalence 
relation that is intended to capture to a reasonable degree the 
intuitive notion that two instruction sequences express the same 
algorithm.
In this section, we give some background on this notion in order to put 
the definition that will be given into context.

In~\cite{BM13d}, where instruction sequences are considered which 
contain backward jump instructions in addition to instructions to set 
and get the content of Boolean registers, forward jump instructions, and 
a termination instruction, it is shown that the function on bit strings 
that models the multiplication of natural numbers on their 
representation in the binary number system can be computed according to 
a minor variant of the long multiplication algorithm by quadratic-length 
instruction sequences without backward jump instructions and by 
linear-length instruction sequences with backward jump instructions.

With that, we implicitly assumed that the instruction sequences without 
backward jump instructions concerned and the instruction sequences with 
backward jump instructions concerned express the same algorithm.
We have asked ourselves the question why this is an acceptable 
assumption and what this says about the notion of an algorithm. 
We considered it an acceptable assumption because all the different 
views on what characterizes an algorithm lead to the conclusion that we 
have to do here with different expressions of the same algorithm.
However, we cannot prove this due to the absence of a mathematically 
precise definition of an equivalence relation on the instruction 
sequences of the kind considered that captures the intuitive notion that 
two instruction sequences express the same algorithm.

The cause of this absence is the general acceptance of the exact 
mathematical concept of a Turing machine and equivalent mathematical 
concepts as adequate replacements of the intuitive concept of an 
algorithm.
Unfortunately, for bit strings of any given length, we can construct at 
least two different Turing machines for the minor variant of the long 
multiplication algorithm referred to above: one without a counterpart of 
a for loop and one with a counterpart of a for loop.
This means that, like programs, Turing machines do not enforce a level 
of abstraction that is sufficient for algorithms.
Moreover, Turing machines are quite remote from anything related to 
actual programming.
Therefore, we doubt whether the mathematical concept of a Turing machine 
is an adequate replacement of the intuitive concept of an algorithm.
This means that we consider a generally accepted mathematically precise 
definition of the concept of an algorithm still desirable.

The existing viewpoints on what is an algorithm are diverse in 
character.
The viewpoint that algorithms are equivalence classes of programs was 
already taken in~\cite{Mil71a}.
This viewpoint was recently also taken in~\cite{Yan11a}, but a rather 
strange twist is that constructions of primitive recursive functions are 
considered to be programs.
In~\cite{Mos01b}, algorithms are viewed as isomorphism classes of tuples 
of recursive functionals that can be defined by repeated application of 
certain schemes.
This viewpoint is somewhat reminiscent of the viewpoint taken 
in~\cite{Yan11a}.
On the other hand, in~\cite{BC82a}, which is concerned with algorithms 
on Kahn-Plotkin's concrete data structures, algorithms are viewed as 
pairs of a function and a computation strategy that resolves choices 
between possible ways of computing the function.
This viewpoint is quite different from the other viewpoints mentioned
above.

In~\cite{KU58a}, it is claimed that the only algorithms are those 
expressed by Kolmogorov machines and that therefore the concept of a 
Kolmogorov machine can be regarded as an adequate formal 
characterization of the concept of an algorithm (see also~\cite{US81a}).
With this the concept of a Kolmogorov machine is actually qualified as a 
replacement of the concept of an algorithm.
In~\cite{Gur00a}, an algorithm is defined as an object that satisfy 
certain postulates.
The postulates concerned seem to be devised with the purpose that 
Gurevich's abstract state machines would satisfy them.
Be that as it may, they are primarily postulates for models of 
computation of a certain kind, i.e.\ replacements of the concept of an 
algorithm.

In~\cite{BDG09a}, it is argued that the intuitive notion that two 
programs express the same algorithm cannot be captured by an equivalence 
relation.
This is also ar\-gued in the philosophical discussion of the view that
algorithms are mathematical objects presented in~\cite{Dea07a}. 
Quite a few of the given arguments are biased towards current patterns 
of thinking within subfields of theoretical computer science like the 
analysis of algorithms and computational complexity theory. 
An important such pattern is the following: if we have proved a result 
concerning programs or abstract machines, then we may formulate it as a 
result concerning algorithms.
This pattern yields, among other things, a biased view on what are the 
properties that two programs expressing the same algorithm must have in 
common.

Moreover, the existence of different opinions and subjective judgments 
concerning the question whether two programs express the same algorithm 
is also a weak argument.
Different opinions and subjective judgments are inevitable in the 
absence of a mathematically precise definition of an equivalence 
relation that captures the intuitive notion that two programs express 
the same algorithm.
All this means that the arguments given in~\cite{BDG09a,Dea07a} are no 
reason for us to doubt the usefulness of looking for an equivalence 
relation that captures to a reasonable degree the intuitive notion that 
two instruction sequences express the same algorithm.  

\section{Intuition about the Notion of Algorithmic Sameness}
\label{sect-intuition-algeqv}

In this section, we give a picture of our intuition about the notion 
that two instruction sequences express the same algorithm.
In section~\ref{sect-definition-saeqv}, we will define in a 
mathematically precise way an equivalence relation corresponding to this 
intuition.

We would like to ground our intuition in the general thinking on the
concept of an algorithm.
However, because the existing viewpoints on what is an algorithm are 
diverse in character and leave loose ends, there is little that we can 
build on.
Therefore, we restrict ourselves to what is virtually the simplest case.
That is, we take two fixed but arbitrary natural numbers $n,m$ and 
restrict ourselves to instruction sequences for computing partial 
functions from $\set{0,1}^n$ to $\set{0,1}^m$.%
\footnote
{We regard total functions as special cases of partial functions.
 Henceforth, total functions are shortly called functions.}
If $n$ and $m$ are very large, this simple case covers, at a very low 
level of program and data representation, many non-interactive programs 
that are found in actual practice.
This simple case has the advantage that data representation is hardly an 
issue in the expressions of algorithms.

In the case that we restrict ourselves to instruction sequences for 
computing partial functions from $\set{0,1}^n$ to $\set{0,1}^m$, taking 
into account the experience gained in~\cite{BM13c,BM13b,BM13d} with 
expressing algorithms by instruction sequences, we consider the 
following to be a first rough approximation of a definition of the 
concept of an algorithm: ``an algorithm is an equivalence class of 
instruction sequences from $\ISbr{n}{m}$ with respect to an equivalence 
relation that completely captures the intuitive notion that two 
instruction sequences express the same algorithm''.
However, the equivalence relation to be defined in 
Section~\ref{sect-definition-saeqv} is likely to incompletely capture 
this notion of algorithmic sameness.
Because we want to capture it to a reasonable degree, we have made a
serious attempt to establish its main characteristics in the simple case 
under consideration.
The main characteristics found are: 
\begin{itemize}
\item
each instruction sequence is algorithmically the same as each 
instruction sequence that produces the same behaviour under execution;
\item
each instruction sequence is algorithmically the same as the instruction 
sequence obtained from it by consistently exchanging $0$ and $1$ as far
as an auxiliary Boolean register is concerned;
\item
each instruction sequence is algorithmically the same as each 
instruction sequence obtained from it by renumbering the auxiliary 
Boolean registers~used;
\item
each instruction sequence is algorithmically the same as each
instruction sequence obtained from it by transposing the basic 
instructions in two always successively executed primitive instructions 
if their combined effect never depends on the order in which they are 
executed.
\end{itemize}
  
The first characteristic expresses that algorithmic sameness is implied
by behav\-ioural equivalence.
It is customary to ascribe this characteristic to the notion of 
algorithmic sameness.
The second characteristic expresses that algorithmic sameness is implied 
by mutual step-by-step simulation of behaviour with $0$ represented by 
$1$ and $1$ represented by $0$ as far as intermediate results are 
concerned.
It is customary to ascribe a characteristic of which this one is a 
special case to the notion of algorithmic sameness.
Usually, the characteristic has a more general form because the data of 
interest are not restricted to bit strings.

The third characteristic can be paraphrased as follows: algorithmic 
sameness identifies two instruction sequences if they only differ in the 
choice of Boolean registers used for storing the different intermediate 
results.
It is customary to ascribe a characteristic of which this one is a 
special case to the notion of algorithmic sameness in those cases where 
the programs concerned are of a concrete form.
The fourth characteristic can be paraphrased as follows: algorithmic 
sameness identifies two instruction sequences if they both express the 
same parallel algorithm.
It should be fairly customary to ascribe this characteristic to the 
notion of algorithmic sameness, but to our knowledge it is seldom made 
explicit.
A remarkable exception is~\cite{MW00a}, a book that describes a 
theoretical framework for optimization of sequential programs by 
parallelization.
 
We remark that it is not easy to think of what may be additional main 
characteristics of the intuitive notion of algorithmic sameness for 
instruction sequences of the kind considered here.
Because we restrict ourselves to algorithms for computing partial 
functions from $\set{0,1}^n$ to $\set{0,1}^m$, for fixed $n,m$, and 
simple instruction sequences without advanced features, it is not 
difficult to gain a comprehensive view of what is within the bounds of 
the possible with regard to the notion of algorithmic sameness.

However, there remain doubtful cases.
An instruction sequence may contain primitive instructions that are
superfluous in the sense that their execution cannot contribute to the 
partial function that it computes.
For example, if a primitive instruction of a form other than 
$\outbr{i}.\setbr{b}$ or $\fjmp{l}$ with $l \neq 1$ is immediately 
followed by a termination instruction, then the former instruction 
cannot contribute to the partial function that the instruction sequence 
computes. 
We are doubtful whether an instruction sequence that contains such a 
superfluous instruction is algorithmically the same as the instruction 
sequence obtained from it by replacing the superfluous instruction by 
$\fjmp{1}$. 
To our knowledge, this does not correspond to any characteristic ever
ascribed to the notion of algorithmic sameness.

As mentioned before, we restrict ourselves in this paper to algorithms 
for computing partial functions from $\set{0,1}^n$ to $\set{0,1}^m$, for 
fixed $n,m$.
We could restrict ourselves further to algorithms for computing partial 
functions from $\set{0,1}^n$ to $\set{0,1}^m$ that can handle their 
restriction to \smash{$\set{0,1}^k$} for each $k < n$ if sufficiently 
many leading zeros are added.
This means that an instruction sequence that computes a partial function 
from $\set{0,1}^n$ to $\set{0,1}^m$ can also be used to compute its 
restriction to \smash{$\set{0,1}^k$} for each $k < n$.
These functions include, for instance, all functions that model the
restriction of an operation on natural numbers to the interval
$[0,2^n - 1]$ on their representation in the binary number system.

\section{The Structural Algorithmic Equivalence Relation}
\label{sect-definition-saeqv}

In this section, we define an algorithmic equivalence relation that
corresponds to the intuition about the notion of algorithmic sameness 
described in Section~\ref{sect-intuition-algeqv}.
Preceding that, we introduce the way in which we will characterize 
instruction sequences several times in this section.

For each of the main characteristics of the intuitive notion of 
algorithmic sameness mentioned in Section~\ref{sect-intuition-algeqv}, 
there is a corresponding equivalence relation which partially captures 
the notion of algorithmic sameness. 
Because some of these equivalence relations may well be interesting as
they are, we first define them and then define an algorithmic 
equivalence relation in terms of them.

Below we define the equivalence relation $\beqv$ on $\ISbr{n}{m}$.
This relation associates each instruction sequence from $\ISbr{n}{m}$ 
with the instruction sequences from $\ISbr{n}{m}$ that produce the same 
behaviour under execution.
The \emph{behavioural equivalence} relation $\beqv$ on $\ISbr{n}{m}$ is 
defined by
\begin{ldispl}
X \beqv Y \Liff \extr{X} = \extr{Y}\;.
\end{ldispl}%

Below we define the equivalence relation $\xeqv$ on $\ISbr{n}{m}$.
This relation associates each instruction sequence from $\ISbr{n}{m}$ 
with the instruction sequences from $\ISbr{n}{m}$ obtained from it by 
consistently exchanging $0$ and $1$ as far as an auxiliary Boolean 
register is concerned.
The \emph{equivalence under bit exchange} relation $\xeqv$ on 
$\ISbr{n}{m}$ is defined as the smallest relation on $\ISbr{n}{m}$ such 
that for all $X,Y \in \ISbr{n}{m}$:
\begin{itemize}
\item
if there exists a finite $I \subset \Natpos$ such that 
$\chi^{}_I(X) = Y$, then $X \xeqv Y$;
\end{itemize}
\sloppy
where, for each finite $I \subset \Natpos$, $\chi^{}_I$ is the unique 
function on $\ISbr{n}{m}$ such that $\len(\chi^{}_I(X)) = \len(X)$ and 
$i_n(\chi^{}_I(X)) = \chi'_I(i_n(X))$ for all $n > 0$, where the 
function $\chi'_I$ on $\PIbr{n}{m}$ is defined as follows: 
\begin{ldispl}
\begin{gceqns}
\chi'_I(f.m)        = f.\chi''(m)\;
 & \mif f \in \set{\auxbr{i} \where i \in I}\;, \\
\chi'_I(f.m)        = f.m\;       
 & \mif f \notin \set{\auxbr{i} \where i \in I}\;, 
\end{gceqns}
\end{ldispl}%
\begin{ldispl}
\begin{gceqns}
\chi'_I(\ptst{f.m}) = \ntst{f.\chi''(m)}\;
 & \mif f \in \set{\auxbr{i} \where i \in I}\;, \\
\chi'_I(\ptst{f.m}) = \ptst{f.m}\;
 & \mif f \notin \set{\auxbr{i} \where i \in I}\;, \\
\chi'_I(\ntst{f.m}) = \ptst{f.\chi''(m)}\;
 & \mif f \in \set{\auxbr{i} \where i \in I}\;, \\
\chi'_I(\ntst{f.m}) = \ntst{f.m}\;
 & \mif f \notin \set{\auxbr{i} \where i \in I}\;, \\
\chi'_I(\fjmp{l}) = \fjmp{l}\;,         \\
\chi'_I(\halt) = \halt\;,               
\end{gceqns}
\end{ldispl}%
where the function $\chi''$ on $\Methbr$ is defined as follows:
\begin{ldispl}
\begin{geqns}
\chi''(\setbr{0}) = \setbr{1}\;, \\
\chi''(\setbr{1}) = \setbr{0}\;, \\ 
\chi''(\getbr) = \getbr\;.
\end{geqns}
\end{ldispl}%

Below we define the equivalence relation $\reqv$ on $\ISbr{n}{m}$.
This relation asso\-ciates each instruction sequence from $\ISbr{n}{m}$ 
with the instruction sequences from $\ISbr{n}{m}$ obtained from it by 
renumbering the auxiliary Boolean registers.
The \emph{equivalence under register renumbering} relation $\reqv$ on 
$\ISbr{n}{m}$ is defined as the smallest relation on $\ISbr{n}{m}$ such 
that for all $X,Y \in \ISbr{n}{m}$:
\begin{itemize}
\item
if there exists a bijection $r$ on $\Natpos$ such that $\rho_r(X) = Y$, 
then $X \reqv Y$;
\end{itemize}
where, for each bijection $r$ on $\Natpos$, the function $\rho_r$ is the 
unique function on $\ISbr{n}{m}$ such that $\len(\rho_r(X)) = \len(X)$ 
and $i_n(\rho_r(X)) = \rho_r'(i_n(X))$ 
for all $n > 0$, where the function $\rho_r'$ on $\PIbr{n}{m}$ is 
defined as follows:
\begin{ldispl}
\begin{geqns}
\rho_r'(f.m) = \rho_r''(f).m\;, \\
\rho_r'(\ptst{f.m}) = \ptst{\rho_r''(f).m}\;, \\
\rho_r'(\ntst{f.m}) = \ntst{\rho_r''(f).m}\;, \\
\rho_r'(\fjmp{l}) = \fjmp{l}\;,         \\
\rho_r'(\halt) = \halt\;,               
\end{geqns}
\end{ldispl}%
where the function $\rho_r''$ on $\Focibr{n}{m}$ is defined as follows:
\begin{ldispl}
\begin{geqns}
\end{geqns}
\qquad
\begin{geqns}
\rho_r''(\inbr{i}) = \inbr{i}\;, \\
\rho_r''(\auxbr{i}) = \auxbr{r(i)}\;, \\ 
\rho_r''(\outbr{i}) = \outbr{i}\;.
\end{geqns}
\end{ldispl}%

Below we define the equivalence relation $\teqv$ on $\ISbr{n}{m}$.
This relation asso\-ciates each instruction sequence from $\ISbr{n}{m}$ 
with the instruction sequences from $\ISbr{n}{m}$ obtained from it by 
transposing the basic instructions in two always suc\-cessively executed 
primitive instructions if their combined effect never depends on the 
order in which they are executed.
The \emph{equivalence under instruction transposition} relation $\teqv$ 
on $\ISbr{n}{m}$ is defined by
\begin{ldispl}
X \teqv Y \Liff \extr{X} \teqvp \extr{Y}\;,
\end{ldispl}%
where $\teqvp$ is the smallest relation on $\Tbr{n}{m}$ such that for 
$x,y,z \in \Tbr{n}{m}$ and $a,b \in \BIbr{n}{m}$:
\begin{itemize}
\item
if $\nm{focus}(a) \neq \nm{focus}(b)$,
then 
$\pcc{(\pcc{x}{b}{y})}{a}{(\pcc{x'}{b}{y'})} \teqvp 
 \pcc{(\pcc{x}{a}{x'})}{b}{(\pcc{y}{a}{y'})}$;%
\footnote
{Here we write $\nm{focus}(a)$ for the unique 
 $f \in \Focibr{n}{m}$ for which there exists an $m \in \Methbr$ such
 that $a = f.m$.}
\item
$x \teqvp x$;
\item
if $x \teqvp y$ and $y \teqvp z$, then $x \teqvp z$;
\item
if $x \teqvp x'$ and $y \teqvp y'$, then 
$\pcc{x}{a}{y} \teqvp \pcc{x'}{a}{y'}$.
\end{itemize}

It is easy to check that the relations $\beqv$, $\xeqv$, $\reqv$, and 
$\teqv$ are actually equivalence relations.

Now we are ready to define the algorithmic equivalence relation 
on $\ISbr{n}{m}$ that corresponds to the intuition about the 
notion of algorithmic sameness described in 
Section~\ref{sect-intuition-algeqv} in terms of the equivalence 
relations $\beqv$, $\xeqv$, $\reqv$, and $\teqv$. 
The \emph{structural algorithmic equivalence} relation $\saeqv$ on 
$\ISbr{n}{m}$ is defined as the smallest relation on $\ISbr{n}{m}$ such 
that for all $X,Y,Z \in \ISbr{n}{m}$:
\begin{itemize}
\item
if $X \beqv Y$ or $X \xeqv Y$ or $X \reqv Y$ or $X \teqv Y$, then 
$X \saeqv Y$;
\item
if $X \saeqv Y$ and $Y \saeqv Z$, then $X \saeqv Z$.
\end{itemize}

It is easy to check that the relation $\saeqv$ is actually an 
equivalence relation.

If $X$ computes a partial function from $\set{0,1}^n$ to $\set{0,1}^m$ 
and $X \saeqv Y$, then $X$ and $Y$ compute the same partial function in
the same number of steps. 
This is made precise in the following theorem.
\begin{theorem}
\label{theorem-saeqv}
For all $X,Y \in \ISbr{n}{m}$, $X \saeqv Y$ only if 
there exists an $l \in \Nat$ such that, 
for all $b'_1,\ldots,b'_l \in \Bool$,
there exist $b''_1,\ldots,b''_l \in \Bool$ such that,
for all $b_1,\ldots,b_n \in \Bool$:
\begin{ldispl}
(\extr{X} \sfause
  ((\Sfcomp{i=1}{n} \inbr{i}.\BR_{b_i}) \sfcomp
   (\Sfcomp{i=1}{l} \auxbr{i}.\BR_{b'_i}))) \sfapply
  (\Sfcomp{i=1}{m} \outbr{i}.\BR_{\False}) 
\\ \;\;\; {} = 
(\extr{Y} \sfause
  ((\Sfcomp{i=1}{n} \inbr{i}.\BR_{b_i}) \sfcomp
   (\Sfcomp{i=1}{l} \auxbr{i}.\BR_{b''_i}))) \sfapply
  (\Sfcomp{i=1}{m} \outbr{i}.\BR_{\False})
\end{ldispl}%
and
\begin{ldispl}
\depth
(\extr{X} \sfpause
  ((\Sfcomp{i=1}{n} \inbr{i}.\BR_{b_i}) \sfcomp
   (\Sfcomp{i=1}{l} \auxbr{i}.\BR_{b'_i})) \sfpause
   (\Sfcomp{i=1}{m} \outbr{i}.\BR_{\False}))
\\ \;\;\; {} = 
\depth
(\extr{Y} \sfpause
  ((\Sfcomp{i=1}{n} \inbr{i}.\BR_{b_i}) \sfcomp
   (\Sfcomp{i=1}{l} \auxbr{i}.\BR_{b''_i})) \sfpause
   (\Sfcomp{i=1}{m} \outbr{i}.\BR_{\False}))\;.
\end{ldispl}%
\end{theorem}
\begin{proof}
By the definition of $\saeqv$ and elementary logical reasoning rules, 
it is sufficient to prove the theorem with $\saeqv$ replaced by $\beqv$, 
$\xeqv$, $\reqv$, and $\teqv$.
The case where $\saeqv$ is replaced by $\beqv$ is trivial.

The case where $\saeqv$ is replaced by $\xeqv$ is proved as outlined 
below.
Let $I \subset \Natpos$ be such that $\chi^{}_I(X) = Y$ (such an $I$ 
exists according to the definition of $\xeqv$).
Then, it is sufficient to prove that $X \xeqv Y$ only if there exists 
an $l \in \Nat$ such that, 
for all $b'_1,\ldots,b'_l \in \Bool$, 
for all $b_1,\ldots,b_n \in \Bool$:
\begin{ldispl}
(\extr{X} \sfause
  ((\Sfcomp{i=1}{n} \inbr{i}.\BR_{b_i}) \sfcomp
   (\Sfcomp{i=1}{l} \auxbr{i}.\BR_{b'_i})))
\\ \;\;\; {} = 
(\extr{Y} \sfause
  ((\Sfcomp{i=1}{n} \inbr{i}.\BR_{b_i}) \sfcomp
   (\Sfcomp{i=1}{l} \auxbr{i}.\BR_{b''_i})))
\end{ldispl}%
and
\begin{ldispl}
\depth
(\extr{X} \sfpause
  ((\Sfcomp{i=1}{n} \inbr{i}.\BR_{b_i}) \sfcomp
   (\Sfcomp{i=1}{l} \auxbr{i}.\BR_{b'_i})))
\\ \;\;\; {} = 
\depth
(\extr{Y} \sfpause
  ((\Sfcomp{i=1}{n} \inbr{i}.\BR_{b_i}) \sfcomp
   (\Sfcomp{i=1}{l} \auxbr{i}.\BR_{b''_i})))\;,
\end{ldispl}%
where, for each $i$ with $1 \leq i \leq l$, 
$b''_i = \overline{b'_i}$ if $i \in I$ and 
$b''_i = b'_i$ if $i \notin I$.%
\footnote
{Here, we write as usual $\overline{b}$ for the complement of $b$.}
This is easily proved by induction on $\len(X)$ and case distinction 
on the possible forms of the first primitive instruction in $X$. 
The case where $\saeqv$ is replaced by $\reqv$ is proved similarly. 

The case where $\saeqv$ is replaced by $\teqv$ is easily proved by
induction on the construction of $\teqvp$. 
\qed
\end{proof}
In Theorem~\ref{theorem-saeqv}, ``only if'' cannot be replaced by ``if 
and only if''.
Take as an example:
\begin{ldispl} 
X = \outbr{1}.\setbr{0} \conc \ldots \conc \outbr{m}.\setbr{0} \conc
     \auxbr{1}.\setbr{0} \conc \halt\;, \\
Y = \outbr{1}.\setbr{0} \conc \ldots \conc \outbr{m}.\setbr{0} \conc
     \auxbr{1}.\setbr{1} \conc \halt\;.
\end{ldispl}%
Then we do not have that $X \saeqv Y$, but $X$ and $Y$ compute the same 
function from $\set{0,1}^n$ to $\set{0,1}^m$ in the same number of 
steps.

Below we define basic algorithms as equivalence classes of instruction 
sequences with respect to the algorithmic equivalence relation $\saeqv$ 
defined above.
Because it is quite possible that $\saeqv$ does not completely capture 
the intuitive notion that two instruction sequences express the same
algorithm, the concept of a basic algorithm introduced below is 
considered to be merely a reasonable approximation of the intuitive 
concept of an algorithm.
The prefix ``basic'' is used because we will introduce in 
Section~\ref{sect-higher-level-algeqv} the concept of a $\PN$-oriented 
algorithm (a concept parameterized by a program notation $\PN$) and that 
concept is basically built on the one introduced here.

Let $\pfunct{f}{\set{0,1}^n}{\set{0,1}^m}$.
Then a \emph{basic algorithm for} $f$ is an 
$A \in \ISbr{n}{m} / {\saeqv}$ such that, for all $X \in A$, $X$ 
computes $f$.%
\footnote
{Here, we write as usual $\ISbr{n}{m} / {\saeqv}$ for the quotient set 
of $\ISbr{n}{m}$ by $\saeqv$.}
A \emph{basic algorithm} is an $A \in \ISbr{n}{m} / {\saeqv}$ for which 
there exists an $\pfunct{f}{\set{0,1}^n}{\set{0,1}^m}$ such that $A$ is
a basic algorithm for $f$.
Let $A$ be a basic algorithm and $X \in \ISbr{n}{m}$.
Then we say that $X$ \emph{expresses} $A$ if $X \in A$.

\section{The Structural Computational Equivalence Relation}
\label{sect-sceqv}

In this section, we define an equivalence relation on $\ISbr{n}{m}$ 
that is coarser than the structural algorithmic equivalence relation 
defined in Section~\ref{sect-definition-saeqv}.
The equivalence relation in question is called the structural 
computational equivalence relation on $\ISbr{n}{m}$.
Although it was devised hoping that it would capture the intuitive 
notion of algorithmic sameness to a higher degree than the structural 
algorithmic equivalence relation, this coarser equivalence relation 
turns out to make equivalent instruction sequences of which it is 
inconceivable that they are considered to express the same algorithm.
Therefore, any equivalence relation that captures the notion of 
algorithmic sameness to a higher degree than the structural algorithmic 
equivalence relation must be finer than the structural computational 
equivalence relation defined in this section.

The structural computational equivalence relation on $\ISbr{n}{m}$ will 
be defined in terms of the equivalence relations $\xeqv$, $\reqv$, and 
$\teqv$ defined in Section~\ref{sect-definition-saeqv} and an 
equivalence relation $\cteqv$ replacing $\beqv$.

Below we define the equivalence relation $\cteqv$ on $\ISbr{n}{m}$.
This relation associates each instruction sequence from $\ISbr{n}{m}$
with the instruction sequences from $\ISbr{n}{m}$ that produce the same
behaviour under execution for all possible contents of the input Boolean
registers.
The \emph{computational trace equivalence} relation $\cteqv$ on 
$\ISbr{n}{m}$ is defined by
\begin{ldispl}
X \cteqv Y
\\ \quad {} \Liff
\Forall{b_1,\ldots,b_n \in \Bool}
 {\extr{X} \sfpause (\Sfcomp{i=1}{n} \inbr{i}.\BR_{b_i})  =
  \extr{Y} \sfpause (\Sfcomp{i=1}{n} \inbr{i}.\BR_{b_i})}\;.
\end{ldispl}%

It is easy to check that the relation $\cteqv$ is actually an
equivalence relation.

Now we are ready to define the equivalence relation $\sceqv$ on 
$\ISbr{n}{m}$ in terms of the equivalence relations $\cteqv$, $\xeqv$, 
$\reqv$, and $\teqv$. 
The \emph{structural computational equivalence} relation $\sceqv$ on
$\ISbr{n}{m}$ is defined as the smallest relation on $\ISbr{n}{m}$ such
that for all $X,Y,Z \in \ISbr{n}{m}$:
\begin{itemize}
\item
if $X \cteqv Y$ or $X \xeqv Y$ or $X \reqv Y$ or $X \teqv Y$, then
$X \sceqv Y$;
\item
if $X \sceqv Y$ and $Y \sceqv Z$, then $X \sceqv Z$.
\end{itemize}

\begin{theorem}
\label{theorem-saeqv-sceqv}
${\saeqv} \subset {\sceqv}$.
\end{theorem}
\begin{proof}
Let $\varphi$ be a proposition containing proposition variables 
$v_1,\ldots,v_n$.
Let $\funct{f}{\set{\False,\True}^n}{\set{\False,\True}^m}$ be such 
that $f(b_1,\ldots,b_n) = \True^m$ if $\varphi$ is satisfied by the 
valuation that assigns $b_1$ to $v_1$, \ldots, $b_n$ to $v_n$ and
$f(b_1,\ldots,b_n) = \False^m$ otherwise.
Let $X \in \ISbr{n}{m}$ be such that $X$ computes $f$.
Let $X' \in \ISbr{n}{m}$ be obtained from $X$ by replacing, for each
$f \in \Fociout{m}$, all occurrences of the basic instruction 
$f.\setbr{\True}$ by the basic instruction $f.\setbr{\False}$.
It follows immediately that not $X \saeqv X'$.
Now suppose that $\varphi$ is not satisfiable.
Then $X \sceqv X'$.
Hence, ${\saeqv} \subset {\sceqv}$.
\qed
\end{proof}
The proof of Theorem~\ref{theorem-saeqv-sceqv} does not only show that
there exist $X$ and $X'$ such that $X \sceqv X'$ and not $X \saeqv X'$.
It also shows that there exist $X$ and $X'$ such that $X \sceqv X'$ 
whereas it is inconceivable that $X$ and $X'$ are considered to express
the same algorithm.  
The point is that an algorithm may take alternatives into account that
will not occur and there is abstracted from such alternatives in the 
case of the structural computational equivalence relation $\sceqv$.

\section{On Algorithmic Equivalence of Higher-Level Programs}
\label{sect-higher-level-algeqv}

In most program notations used for actual programming, programs are more
advanced than the instruction sequences from $\ISbr{n}{m}$.
In this section, we show that the algorithmic equivalence relation on 
$\ISbr{n}{m}$ defined in Section~\ref{sect-definition-saeqv} can easily 
be lifted to programs in a higher-level program notation if the approach 
of projection semantics (see below) is followed in giving the program 
notation concerned semantics.

\PGA\ instruction sequences has not been designed to play a part in 
actual programming. 
In fact, they are less suitable for actual programming than the 
instruction sequences that are found in low-level program notations such 
as assembly languages. 
However, even high-level program notations with advanced features such
as conditional constructs, loop constructs, and subroutines that may 
call themselves recursively can be given semantics by means of a mapping 
from the programs in the notation concerned to \PGA\ instruction 
sequences and a service family.
This approach to the semantics of program notations, which is followed 
in~\cite{BB06a,BL02a}, is called projection semantics. 

We define a program notation as a triple $(L,\varphi,S)$, where $L$ is a 
set of programs, $\varphi$ is a mapping from $L$ to the set of all \PGA\ 
instruction sequences, and $S$ is a service family.
The mapping $\varphi$ is called a projection.
The behaviour of each program $P$ in $L$ is determined by $\varphi$ and 
$S$ as follows: the behaviour of $P$ is the behaviour represented by the 
thread $\extr{\varphi(P)} \sfause S$.

For certain program notations, it is fully sufficient that $S$ is the 
empty service family.
An example is the program notation in which programs are finite 
instruction sequences that differ from finite \PGA\ instruction 
sequences in that they may contain backward jump instructions
(see e.g.~\cite{BL02a}).
For certain other program notations, it is virtually or absolutely 
necessary that $S$ is a non-empty service family.
Examples are program notations with features such as subroutines that 
may call themselves recursively (see e.g.~\cite{BB06a}).
It is clear that $(\ISbr{n}{m},\iota,\emptysf)$, where $\iota$ is the 
identity mapping on $\ISbr{n}{m}$, is the proper program notation for 
$\ISbr{n}{m}$.

Because we build on the algorithmic equivalence relation $\saeqv$ on 
$\ISbr{n}{m}$ defined in Section~\ref{sect-definition-saeqv}, we 
restrict ourselves to program notations $(L,\varphi,S)$ that satisfy the 
following conditions:
\begin{itemize}
\item
$\encap{\Focibr{n}{m}}(S) = S$;
\item
for each $P \in L$, there exists an $X \in \ISbr{n}{m}$ such that 
$\extr{\varphi(P)} \sfause S = \extr{X}$;
\item
for each $X \in \ISbr{n}{m}$, there exists a $P \in L$ such that 
$\extr{\varphi(P)} \sfause S = \extr{X}$.
\end{itemize}
The first condition excludes cases in which $S$ is inappropriate.
The second condition allows for lifting $\saeqv$ from $\ISbr{n}{m}$ to 
$L$ and the third condition allows for doing so such that the coarseness 
of $\saeqv$ is preserved.  
The last two conditions entail a restriction to program notations with 
the same computational power as $(\ISbr{n}{m},\iota,\emptysf)$.
The following explains why $S$ is inappropriate if the first 
condition is not satisfied.
If $S$ contains a service whose focus is from $\Focibr{n}{m}$, which is 
the case if the first condition is not satisfied, the composition of $S$ 
with the family of Boolean register services that is used by an 
instruction sequence from $\ISbr{n}{m}$ in computing a partial function 
from $\set{0,1}^n$ to $\set{0,1}^m$ may lead to an undesirable result by 
axiom SFC4 (Table~\ref{axioms-SFA}).

For each program notation $(L,\varphi,S)$ that satisfies these 
conditions, there exists a program notation $(L,\varphi',\emptyset)$ 
such that, for each $P \in L$,
$\extr{\varphi(P)} \sfause S = \extr{\varphi'(P)} \sfause \emptyset$.
This means that we may restrict ourselves to program notations
$(L,\varphi,S)$ with $S = \emptyset$.
However, this restriction leads in some cases to a complicated 
projection $\varphi$.
Take, for example, a program notation with subroutines that may call 
themselves recursively.
Owing to the restriction to program notations with the same 
computational power as $(\ISbr{n}{m},\iota,\emptysf)$, the recursion
depth is bounded.
Using a service that makes up a bounded stack of natural numbers is 
convenient and explanatory in the description of the behaviour of the 
programs concerned, but it is not necessary.

Below we lift the algorithmic equivalence relation $\saeqv$ defined in 
Section~\ref{sect-definition-saeqv} from $\ISbr{n}{m}$ to programs in a 
higher-level program notation.
The lifted equivalence relation is uniformly defined for all program 
notations.
For each program notation, it captures the notion of algorithmic 
sameness to the same degree as~$\saeqv$.

Let $\PN = (L,\varphi,S)$ be a program notation.
Then we define the \emph{structural algorithmic equivalence} relation 
${\gsaeqv{\PN}} \subseteq L \x L$ \emph{for} $\PN$ as follows:
$P \gsaeqv{\PN} Q$ iff  there exist $X,Y \in \ISbr{n}{m}$ such that 
$\extr{\varphi(P)} \sfause S = \extr{X}$, 
$\extr{\varphi(Q)} \sfause S = \extr{Y}$, and $X \saeqv Y$.

Suppose that the projection $\varphi$ in the program notation 
$\PN = (L,\varphi,S)$ is optimizing in the sense that it removes or 
replaces in whole or in part that which is superfluous.
Because of this, certain programs would be structurally algorithmically 
equivalent that would not be so otherwise.
This may very well be considered undesirable (cf.\ the discussion about 
the replacement of superfluous instructions in 
Section~\ref{sect-intuition-algeqv}).
Fortunately, projections like the supposed one are excluded by the third 
condition that must be satisfied by the program notations to which we 
restrict ourselves.

The concept of a basic algorithm can easily be lifted to programs in a
higher-level program notation as well.
The lifted concept is uniformly defined for all program notations.
For each program notation, it is of course still merely an approximation 
of the intuitive concept of an algorithm.

Let $\PN = (L,\varphi,S)$ be a program notation and
let $\pfunct{f}{\set{0,1}^n}{\set{0,1}^m}$.
Then a \emph{$\PN$-oriented algorithm for} $f$ is an 
$A \in L / {\gsaeqv{\PN}}$ such that 
$\set{\extr{\varphi(P)} \sfause S \where P \in A}$ is a basic algorithm
for $f$.
A \emph{$\PN$-oriented algorithm} is an $A \in L / {\gsaeqv{\PN}}$ for 
which there exists an $\pfunct{f}{\set{0,1}^n}{\set{0,1}^m}$ such that 
$A$ is a $\PN$-oriented algorithm for~$f$.

Let $\PN = (L,\varphi,S)$ be a program notation.
Then the concept of a $\PN$-oriented algorithm is essentially the same 
as the concept of a basic algorithm in the sense that
there exists a surjection $\psi$ from $L$ to $\ISbr{n}{m}$ such that, 
for all $A \in L / {\gsaeqv{\PN}}$, $A$ is a $\PN$-oriented algorithm 
iff $\set{\psi(P) \where P \in A}$ is a basic~algorithm.

In~\cite{BM13d}, instruction sequences without backward jump 
instructions and instruction sequences with backward jump instructions 
were given which were assumed to express the same minor variant of the 
long multiplication algorithm.
If we take the program notation used as $\PN$, then according to the 
definition of $\gsaeqv{\PN}$ given above, the instruction 
sequences without backward jump instructions and the instruction 
sequences with backward jump instructions are structurally 
algorithmically equivalent, which indicates that they express the same 
algorithm.

Two or more different program notations as considered in this section 
can be combined into one.
This is done in the obvious way if the sets of programs to be combined 
are mutually disjoint and the sets of foci that serve as names in the 
service families to be combined are mutually disjoint.
Otherwise, sufficient renaming must be applied first.
The case of a structural algorithmic equivalence relation on the 
programs from two or more different program notations is covered by the 
definition given above as well because of this possibility to combine 
several program notations into one.

The definition of the structural algorithmic equivalence relation for a 
program notation $(L,\varphi,S)$ given above is intended to show that a 
workable such relation can be obtained by lifting the structural 
algorithmic equivalence relation $\saeqv$ from $\ISbr{n}{m}$ to $L$.
It should be mentioned that the definition leaves room for further 
investigation because of the following point.
The service family $S$ is primarily meant for dealing with advanced 
control flow features of the program notation, but $S$ can also be used 
for other purposes.
It may be the case that some of these other purposes make certain 
programs structurally algorithmically equivalent according to the 
definition given above whereas it is debatable whether they are so at 
the level of the program notation. 

The projection $\varphi$ of a program notation $(L,\varphi,S)$ can be 
viewed as a theoretical compiler in which practical issues, such as the 
compactness, space efficiency, and time efficiency of the output, are no 
considerations.
We have not pursued the question whether this makes our approach to 
algorithmic equivalence of higher-level programs relevant to the 
construction of correct compilers.

\section{Discussion on What is an Algorithm}
\label{sect-discussion}

In this section, we point out that we are still far from the definitive 
answer to the question ``what is an algorithm?''.

When we consider an $\pfunct{f}{\set{0,1}^n}{\set{0,1}^m}$, we are  
faced with the question what is an algorithm $A$ for computing $f$. 
We may assume the existence of a class $\ALGO{n}{m}$ of algorithms for 
partial functions  from $\set{0,1}^n$ to $\set{0,1}^m$, in which case 
we may assume that $A \in \ALGO{n}{m}$.
However, as reported in Section~\ref{sect-background-algeqv}, the 
computer science literature gives little to go on with regard to 
$\ALGO{n}{m}$.
What we know about algorithms for computing $f$ is that they can be
expressed by one or more instruction sequences from $\ISbr{n}{m}$.
Moreover, different algorithms for computing $f$ may differ in
comprehensibility and efficiency, i.e.\ the number of steps in which 
they compute $f$ and the number of auxiliary Boolean registers that they 
need to compute $f$. 

Suppose that $\Gamma_A(X)$ is a real number in the interval $[0,1]$ 
which represents the degree to which algorithm $A$ is expressed by 
instruction sequence $X$.
What we can say about $\ALGO{n}{m}$ based on the work presented in this 
paper is that:
\begin{itemize}
\item
$\ALGO{n}{m}$ is approximated by the class of equivalence classes of 
instruction sequences from $\ISbr{n}{m}$ with respect to the structural 
algorithmic equivalence relation $\saeqv$;
\pagebreak[2]
\item
an instruction sequence $X \in \ISbr{n}{m}$ expresses an 
$A \in \ALGO{n}{m}$ if $X \in A$;
\item
$\Gamma_A$ is the characteristic function of $A$.
\end{itemize}
Perhaps we can say more about $\ALGO{n}{m}$ if we allow $\Gamma_A$ to 
yield values other than $0$ and $1$ to deal with uncertainty about 
whether an instruction sequence expresses an algorithm.
This seems to fit in with actual practice where an algorithm $A$ is an 
idea in the mind of a programmer, say $P_1$, and $\Gamma_A$ comprises 
judgments by $P_1$.
If a colleague of $P_1$ tries to get the ``idea of $A$'', then 
questioning $P_1$ about his or her judgments $\Gamma_A$ may be the best 
option available to the colleague.
If the judgments of all members of a group of programmers are the same,
then $A$ is given in the group by the shared $\Gamma_A$.

Returning to the restricted setting considered in this paper, we realize
that important questions are still unanswered.
Among them are:
\begin{itemize}
\item
how can structural algorithmic equivalence be generalized from finite 
\PGA\ instruction sequences to finite and eventually periodic infinite
\PGA\ instruction sequences;
\item
how can algorithms be represented directly, rather than indirectly via 
a representative of an equivalence class;
\item
what is the exact connection between algorithms and efficiency of 
computation?
\end{itemize}

\section{Concluding Remarks}
\label{sect-concl}

We have looked for an equivalence relation on instruction sequences that 
captures to a reasonable degree the intuitive notion that two 
instruction sequences express the same algorithm. 
Restricting ourselves to algorithms for computing partial functions from 
$\set{0,1}^n$ to $\set{0,1}^m$, for fixed $n,m$, we have pictured our 
intuition about the notion that two instruction sequences express the 
same algorithm, defined an algorithmic equivalence relation 
corresponding to this intuition, and defined the concept of a basic 
algorithm using this equivalence relation.
We have also shown how this algorithmic equivalence relation can be 
lifted to programs in a higher-level program notation, i.e.\ a program
notation with advanced features such as conditional constructs, loop 
constructs, and subroutines that may call themselves recursively.

We have further defined an equivalence relation whose relevance is that 
any equivalence relation that captures the notion that two instruction
sequences express the same algorithm to a higher degree than the 
algorithmic equivalence relation defined in this paper must be finer 
than this equivalence relation.
We have also pointed out that we are still far from the definitive 
answer to the question ``what is an algorithm?''.

We leave it for future work to show how the algorithmic equivalence 
relation defined in this paper can be generalized to the case where 
programs compute partial functions on data of a higher level than bit 
strings.  
In case the usual viewpoint is taken that the data may be differently 
represented in algorithmically equivalent programs, defining such a 
generalization in a mathematically precise way is nontrivial.
The issue is that this viewpoint, although intuitively clear, leaves a 
loose end: it remains vague about which inescapable differences between 
programs due to different data representations must be considered 
inessential for algorithmic equivalence.
It is mainly the tying up of this loose end what makes defining the 
generalization of the algorithmic equivalence relation nontrivial.

In~\cite{Mil71a}, Milner takes the viewpoint that algorithms are 
equivalence classes of programs.
Because Milner seems to be the first person taking this viewpoint, we 
coin the name 
\emph{Milner's algorithmic equivalence hypothesis} for the hypothesis
that there exists an equivalence relation on programs that captures the 
intuitive notion that two programs express the same algorithm.
In some more recent papers, notably~\cite{BDG09a,Dea07a}, other people 
take this hypothesis to be implausible.
In the following paragraph, the connection between Milner's algorithmic 
equivalence hypothesis and the work presented in this paper is 
summarized.

Our point of departure is the following weakening of Milner's 
algorithmic equivalence hypothesis: there exists a family of equivalence 
relations on programs that capture the intuitive notion that two 
programs express the same algorithm to some degree.
Restricting ourselves to programs that are instruction sequences for 
computing partial functions from $\set{0,1}^n$ to $\set{0,1}^m$, we have
defined two equivalence relations on the instruction sequences 
concerned: the structural algorithmic equivalence relation and the
structural computational equivalence relation.
We believe that the structural algorithmic equivalence relation is an 
equivalence relation that belongs to the hypothesized family and that 
any equivalence relation belonging to the hypothesized family must be 
finer than the structural computational equivalence relation.

Note that our weakening of Milner's algorithmic equivalence hypothesis 
is reminiscent of process theory, where different behavioural 
equivalence relations capture the intuitive notion that two processes 
exhibit the same behaviour to different degrees 
(see e.g.~\cite{Gla01a}).

\subsection*{Acknowledgements}

We thank two anonymous referees for carefully reading a preliminary 
version of this paper and for suggesting improvements of the 
presentation of the paper.

\bibliographystyle{splncs03}
\bibliography{IS}

\end{document}